\documentclass[a4paper,11pt,reqno]{amsart}

\usepackage{amsaddr}
\usepackage[margin=1in,includehead,includefoot]{geometry}
\usepackage[utf8]{inputenc}
\usepackage[T1]{fontenc}
\usepackage{lmodern,microtype}
\usepackage{upref}
\usepackage{mathtools,amssymb,amsthm}
\usepackage{dsfont}
\usepackage{hyperref}
\usepackage{algpseudocode}

\makeatletter
\let\old@setaddresses\@setaddresses
\def\@setaddresses{\bigskip{\parindent 0pt\let\scshape\relax\let\ttfamily\relax\old@setaddresses}}
\makeatother

\newtheorem{definition}{Definition}
\newtheorem{theorem}{Theorem}
\newtheorem{lemma}[theorem]{Lemma}

\usepackage[ruled, vlined]{algorithm2e}
\usepackage{wrapfig}

\newcommand{\NP}{\mathsf{NP}}
\newcommand{\NL}{\mathsf{NL}}
\newcommand{\CoNP}{\mathsf{coNP}}
\newcommand{\UP}{\mathsf{UP}}
\newcommand{\CoUP}{\mathsf{coUP}}
\newcommand{\Poly}{\mathsf{P}}
\newcommand{\CLS}{\mathsf{CLS}}
\newcommand{\PLS}{\mathsf{PLS}}
\newcommand{\UEOPL}{\mathsf{UniqueEOPL}}
\newcommand{\PSPACE}{\mathsf{PSPACE}}

\newcommand{\dbar}{\overline{d}}

\newcommand{\dest}{\{d, \, \dbar\}}
\newcommand{\arr}{ARRIVAL}

\newcommand{\N}{\mathbb{N}}

\newcommand{\bigO}{O}
\newcommand{\floor}[1]{\left\lfloor #1 \right\rfloor}
\newcommand{\ceil}[1]{\left\lceil #1 \right\rceil}
\newcommand{\dist}{\mathrm{dist}}

\hypersetup{
  pdftitle={A subexponential algorithm for \arr}
  pdfauthor={Bernd G\"artner, Sebastian Haslebacher, Hung Phuc Hoang}
}

\title{A subexponential algorithm for \arr}

\author{Bernd G\"artner, Sebastian Haslebacher, and Hung P. Hoang}
\address{Department of Computer Science, ETH Z\"urich, Switzerland \\ {\tt \{gaertner, sebastian.haslebacher, hung.hoang\}@inf.ethz.ch}}

\begin{document}

\begin{abstract}
  The \arr\ problem is to decide the fate of a train moving along the
  edges of a directed graph, according to a simple (deterministic)
  pseudorandom walk. The problem is in $\NP\cap\CoNP$ but not known to
  be in $\Poly$. The currently best algorithms have runtime
  $2^{\Theta(n)}$ where $n$ is the number of vertices. This is not
  much better than just performing the pseudorandom walk. We develop a
  subexponential algorithm with runtime $2^{\bigO(\sqrt{n}\log n)}$. We
  also give a polynomial-time algorithm if the graph is almost
  acyclic. Both results are derived from a new general approach to
  decide \arr\ instances.
\end{abstract}


\maketitle

\section{INTRODUCTION}
\label{sec:intro}
Informally, the \arr\ problem is the following (we quote from Dohrau et
al.~\cite{Dohrau2017}):
\begin{quotation}
 Suppose that a train is running along a railway network, starting
  from a designated origin, with the goal of reaching a designated
  destination. The network, however, is of a special nature: every
  time the train traverses a switch, the switch will change its
  position immediately afterwards. Hence, the next time the train
  traverses the same switch, the other direction will be taken, so
  that directions alternate with each traversal of the switch.

  Given a network with origin and destination, what is the
  complexity of deciding whether the train, starting at the origin,
  will eventually reach the destination?
\end{quotation}

\arr\ is arguably the simplest problem in $\NP\cap\CoNP$ that
is not known to be in $\Poly$. Due to its innocence and at the same
time unresolved complexity status, \arr\ has
attracted quite some attention recently. The train run can be
interpreted as a deterministic simulation of a random walk that
replaces random decisions at a switch by perfectly fair
decisions. Such pseudorandom walks have been studied before under the
names of \emph{Eulerian walkers}~\cite{Priezzhev1996},
\emph{rotor-router walks}~\cite{Holroyd2010}, and \emph{Propp
  machines}~\cite{Cooper2007}. The reachability question as well as
$\NP$ and $\CoNP$ membership are due to Dohrau et
al.~\cite{Dohrau2017}.

Viewed somewhat differently, \arr\ is a \emph{zero player game} (a process
that runs without a controller); in contrast, three other well-known
graph games in $\NP\cap\CoNP$ that are not known to be in $\Poly$ are
two-player (involving two controllers). These are \emph{simple
  stochastic games}, \emph{mean-payoff games} and \emph{parity
  games}~\cite{Con,PZ,Jur}. Moreover, it is stated in (or easily seen
from) these papers that the one-player variants (the strategy of one
controller is fixed) have polynomial-time algorithms. In light
of this, one might expect a zero-player game such as \arr\ to be \emph{really}
simple. But so far, no polynomial-time algorithm could be found.

On the positive side, the $\NP\cap\CoNP$ complexity upper bound could
be strengthened in various ways. \arr\ is in $\UP\cap\CoUP$, meaning
that there are efficient verifiers that accept \emph{unique}
proofs~\cite{Gaertner2018}. A search version of \arr\ has been
introduced by Karthik C.~S.\ and shown to be in
$\PLS$~\cite{Karthik2017}, then in $\CLS$~\cite{Gaertner2018}, and
finally in $\UEOPL$~\cite{Gaertner2018, UEOPL}. The latter complexity
class, established by Fearnley et al.~\cite{UEOPL}, has an intriguing
complete problem, but there is no evidence that \arr\ is complete for
$\UEOPL$.

Concerning complexity lower bounds, there is one result: \arr\ is
$\NL$-hard~\cite{Fearnley2018}. This is not a very strong statement
and means that every problem that can be solved by a nondeterministic
log-space Turing machine reduces (in log-space) to \arr.

Much more interesting are the natural one- and two-player variants of
\arr\ that have been introduced in the same
paper by Fearnley et al.~\cite{Fearnley2018} and later expanded by 
Ani et al.~\cite{Ani2020}. These variants allow a better comparison with the
previously mentioned graph games. It turns out that the one-player
variants of \arr\ are $\NP$-complete, and that the two-player variants are
$\PSPACE$-hard~\cite{Fearnley2018, Ani2020}. This shows that the $p$-player
variant of \arr\ is probably strictly harder than the $p$-player
variants of the other graph games mentioned before, for $p=1,2$. This makes it a bit less surprising that \arr\ itself ($p=0$) could so far not be shown to lie in $\Poly$.

On the algorithmic side, the benchmark is the obvious algorithm for
solving \arr\ on a graph with $n$ vertices: simulate the train
run. This is known to take at most $\bigO(n2^n)$ steps (after this, we can
conclude that the train runs forever)~\cite{Dohrau2017}. There is also
an $\Omega(2^n)$ lower bound for the simulation~\cite{Dohrau2017}. The
upper bound was improved to $\bigO(p(n) 2^{n/2})$ (in expectation) for
some polynomial $p$, using a way to efficiently sample from the
run~\cite{Gaertner2018}. The same bound was later achieved
deterministically~\cite{Hung,Rote}, and the approach can be
refined to yield a runtime of $\bigO(p(n) 2^{n/3})$, the currently
best one for general \arr\ instances~\cite{Rote}.

In this paper, we prove that \arr\ can be decided in subexponential time
$2^{\bigO(\sqrt{n}\log n)}$. While this is still far away from the desired
polynomial-time algorithm, the new upper bound is making the first significant
progress on the runtime. We also prove that polynomial runtime can be
achieved if the graph is close to acyclic, meaning that it can be made
acyclic by removing a constant number of vertices.

As the main technical tool from which we derive both results, we
introduce a generalization of \arr. In this \emph{multi-run} variant,
there is a subset $S$ of vertices where additional trains may start
and also terminate. It turns out that if we start the right numbers of
trains from the vertices in $S$, we also decide the original instance,
so the problem is reduced to searching for these right numbers. We
show that this search problem is well-behaved and can be solved by
systematic guessing, where the number of guesses is exponential
in $|S|$, not in $n$.

We are thus interested in cases where $S$ is small but at the same time
allows a sufficiently fast evaluation of a given guess. For the
subexponential algorithm, we choose $S$ as a set of size
$\bigO(\sqrt{n})$, with the property that a train can only take a
subexponential number of steps until it terminates (in $S$ or a
destination). For almost acyclic graphs, we choose $S$ as a minimum
feedback vertex set, a set whose removal makes the graph acyclic. In
this case, a train can visit any vertex only once before it
terminates.

The multi-run variant itself is an interesting new approach 
to the \arr\ problem, and other applications of it might be found in
the future.

\section{\arr} 
\label{sec:preliminaries}

The \arr\ problem was introduced by Dohrau et al.~\cite{Dohrau2017} as
the problem of deciding whether the train arrives at a given
destination or runs forever. Here, we work in a different but
equivalent setting (implicitly established by Dohrau et al.\ already)
in which the train always arrives at one of two destinations, and we
have to decide at which one. The definitions and results from Dohrau
et al.~\cite{Dohrau2017} easily adapt to our setting. We still provide
independent proofs, derived from the more general setting that we
introduce in Section~\ref{sec:framework}.

Given a finite set of vertices $V$, an origin $o\in V$, two
destinations $d, \, \dbar\notin V$ and two functions
$s_{even}, \, s_{odd} : V \rightarrow V \cup \dest$, the 6-tuple
$A = (V, \, o, \, d, \, \dbar, \, s_{even}, \, s_{odd})$ is an
\emph{\arr\ instance}. The vertices $s_{even}(v)$ and $s_{odd}(v)$ are
called the even and the odd successor of $v$.

An \arr\ instance $A$ defines a directed graph, connecting
each vertex $v\in V$ to its even and its odd successor. We call this the
\emph{switch graph} of $A$ and denote it by $G(A)$. To avoid special
treatment of the origin later, we introduce an artificial vertex
$Y\notin V\cup\{d, \, \dbar\}$ (think of it as the ``train yard'') that
only connects to the origin $o$. Formally, $G(A)=(V(A),E(A))$ where 
$V(A)=V\cup\{Y, \, d, \, \dbar\}$ and
$E(A) = \{(Y, o)\} \cup \{ (v, \, s_{even}(v)) : v \in V\} \cup \{ (v,
\, s_{odd}(v)) : v \in V\}$. We also refer to $E(A)$ simply as the
edges of $A$. An edge $e\neq(Y,o)$ is called \emph{proper}.

The \emph{run procedure} is the following. For every vertex we
maintain a current and a next successor, initially the even and the
odd one. We put a token (usually referred to as the train) at $o$ and
move it along switch graph edges until it reaches either $d$ or
$\dbar$. Whenever the train is at a vertex $v$, we move it to $v$'s
current successor and then swap the current and the next successor;
see Algorithm~\ref{algo:run_procedure} for a formal description and
Figure~\ref{fig:arrival} for an example.

\begin{algorithm}[htb]
\DontPrintSemicolon
\SetKwRepeat{Do}{do}{while}
\SetKwComment{cmt}{(*}{*)}
\SetNoFillComment 
\KwIn{\arr{} instance  $A=(V, \, o, \,
d, \, \dbar, \, s_{even}, \, s_{odd})$}
\KwOut{destination of the train: either $d$ or $\dbar$}

Let $s_{curr}$ and $s_{next}$ be arrays indexed by the vertices of $V$ \;
\For{$v \in V$}{
	$s_{curr}[v] \gets s_{even}(v)$ \;
	$s_{next}[v] \gets s_{odd}(v)$ \;
}

$v \gets o$ \tcc*[r]{traversal of edge $(Y,o)$}
\While{$v \neq d$ and $v \neq \dbar$}{
	$w \gets s_{curr}[v]$ \;
	swap($s_{curr}[v], \, s_{next}[v])$ \;
	$v \gets w$ \tcc*[r]{traversal of edge $(v,w)$}
}
\Return{$v$}
\caption{Run Procedure}
\label{algo:run_procedure}
\end{algorithm}

Algorithm~\ref{algo:run_procedure} (Run procedure) may cycle, but we
can avoid this by assuming that from every vertex $v\in V$, one of $d$
and $\dbar$ is reachable along a directed path in $G(A)$. We call such
an \arr\ instance \emph{terminating}, since it guarantees that either
$d$ or $\dbar$ is eventually reached.

\begin{figure}[htb]
  \begin{center}
    \includegraphics[width=0.3\textwidth]{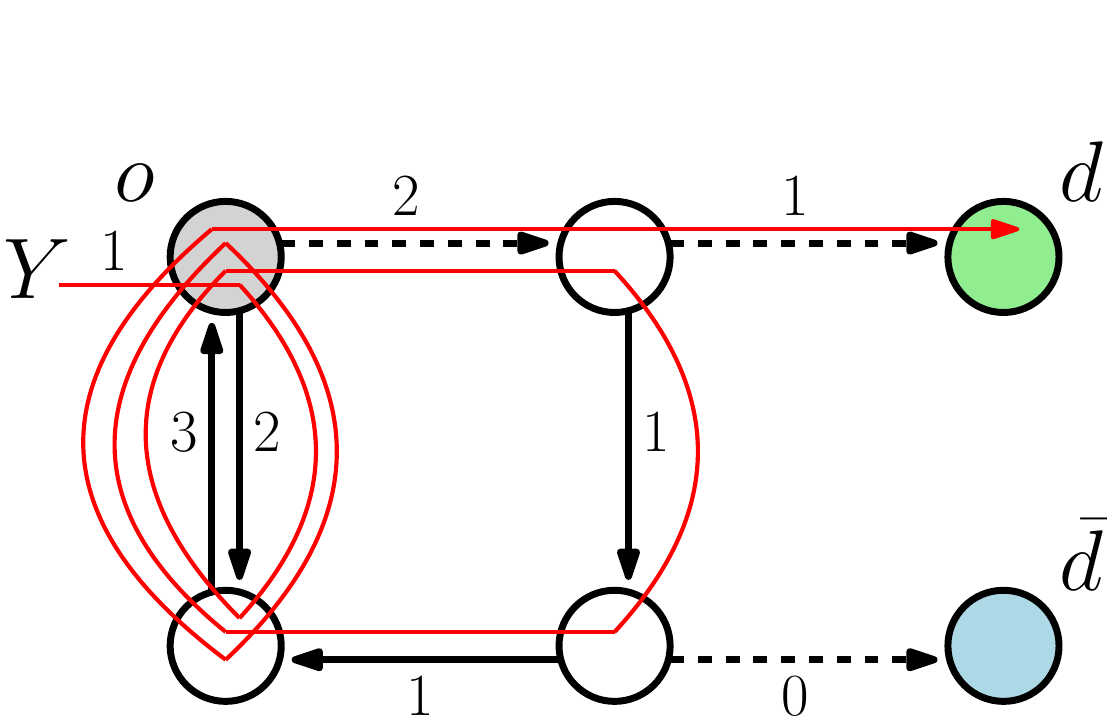}
  \end{center}
  \caption{A terminating \arr\ instance and the train run. Bold edges
    go to the even successors, dashed edges to the odd successors. The
    two successors may coincide (lower left vertex). The numbers
    indicate how often each edge is traversed by the train. \label{fig:arrival}}  
\end{figure}

\begin{lemma}\label{lem:visits}
  Let $A=(V, \, o, \,d, \, \dbar, \, s_{even}, \, s_{odd})$ be a
  terminating \arr{} instance, $|V|=n$. Let $v\in V$ and suppose
  that  the shortest path from $v$ to a destination in $G(A)$
  has length $m$. Then $v$ is visited (the train is at $v$) at most
  $2^m$ times by Algorithm~\ref{algo:run_procedure} (Run Procedure).
\end{lemma}

\begin{proof} Let $v=v_m,v_{m-1},\ldots,v_0\in\dest$ be the sequence
  of vertices on a shortest path from $v$ to
  $\dest$. Consider the first $2^m$ visits to $v$ (if there are
  less, we are done). Once every two consecutive visits, the train
  moves on to $v_{m-1}$, so we can consider the first $2^{m-1}$ visits
  to $v_{m-1}$ and repeat the argument from there to show that $v_i$
  is visited at least $2^i$ times for all $i$, before $v$ exceeds
  $2^m$ visits. In particular, $v_0\in\dest$ is visited, so the run
  indeed terminates within at most $2^m$ visits to $v$.
\end{proof}

\begin{lemma}\label{lem:termination}
  Let $A=(V, \, o, \,d, \, \dbar, \, s_{even}, \, s_{odd})$ be a terminating \arr{} instance, $|V|=n$. 
  Let $\ell$ be the maximum length of the shortest path from a vertex in $V$ to a destination.
  Algorithm~\ref{algo:run_procedure} (Run Procedure) traverses at most $(n - \ell + 2) 2^\ell-2$
  proper edges.
\end{lemma}

\begin{proof}
  By Lemma~\ref{lem:visits}, the total number of visits to vertices
  $v\in V$ is bounded by $\sum_{i=1}^n n_i 2^i$, where $n_i$ is the
  number of vertices with a shortest path of length $i$ to a destination. 
  We have $n_i>0$ if and only if $i\leq\ell$, and hence the sum is maximized 
  if $n_i=1$ for all $i<\ell$, and $n_{\ell}=n-\ell+1$. In this case, the sum is
  $(n - \ell + 2) 2^{\ell}-2$. The number of proper edges being traversed (one after every
  visit of $v\in V$) is the same.
\end{proof}

Given a terminating instance, \arr{} is the problem of deciding
whether Algorithm~\ref{algo:run_procedure} (Run Procedure) returns $d$
(YES instance) or $\dbar$ (NO instance). It is unknown whether
\arr\ $\in \Poly$, but it is in $\NP\cap\CoNP$, due to the existence of
\emph{switching flows} that are certificates for the
output of Algorithm~\ref{algo:run_procedure} (Run Procedure).

\subsection{Switching Flows}

For a vertex $v$ and a set of edges $E$,
we will denote the set of outgoing edges of $v$ by
$E^+(v)$. Analogously, we will denote the set of incoming edges of $v$
by $E^-(v)$. Furthermore, for a function $x : E \rightarrow \N_0$, we
will also use the notation $x_e$ instead of $x(e)$ to denote the
value of $x$ at some edge $e \in E$. Lastly, given some vertex $v$,
edges $E$ and a function $x : E \rightarrow \N_0$, we will use $x^+(v)
:= \sum_{e \in E^+(v)} x_e$ to denote the \emph{outflow} of $x$
at $v$ and $x^-(v) :=\sum_{e \in E^-(v)} x_e$ to denote the
\emph{inflow} of $x$ at $v$. For two functions $x,x': E \rightarrow
\N_0$, we write $x\leq x'$ if this holds componentwise, i.e.\
$x_e\leq x'_e$ for all $e\in E$.

\begin{definition}[Switching Flow~\cite{Dohrau2017}]\label{def:switching_flow}
	Let $A=(V, \, o, \,d, \, \dbar, \, s_{even}, \, s_{odd})$ be a
        terminating \arr{} instance with edges $E$. A function $x
        : E \rightarrow \N_0$ is a \emph{switching flow} for $A$ if
\[
\begin{array}{rccll}
x^+(Y) &=& 1, \\
  x^+(v) - x^-(v) &=& 0, & \quad v\in V & \text{(flow conservation)}\\
  x_{(v, s_{even}(v))} -x_{(v, s_{odd}(v))}&\in&\{0,1\}, & \quad v\in V& \text{(switching behavior)}.
\end{array}
\]
Moreover, $x$ is called a switching flow \emph{to} $t\in\dest$ if $x^-(t)=1$.
\end{definition}

Note that due to flow conservation, a switching flow is a switching 
flow either to $d$ or to $\dbar$: exactly one of the destinations
must absorb the unit of flow emitted by $Y$. If we set $x_e$ to the
number of times the edge $e$ is traversed in
Algorithm~\ref{algo:run_procedure} (Run Procedure), we obtain a
switching flow to the output; see Figure~\ref{fig:arrival} for an
example. Indeed, every time the train enters
$v\in V$, it also leaves it; this yields flow conservation. The strict
alternation between the successors (beginning with the even one)
yields switching behavior.

Hence, the existence of a switching flow to the output is necessary
for obtaining the output. Interestingly, it is also sufficient. For
that, it remains to prove that we cannot have switching flows to both
$d$ and $\dbar$ for the same instance.

\begin{theorem}[Switching flows are certificates~\cite{Dohrau2017}]
\label{theorem:dohrau}
Let $A=(V, \, o, \, d, \, \dbar, \, s_{even}, \, s_{odd})$ be a
terminating \arr{} instance, $t\in\dest$. Algorithm~\ref{algo:run_procedure} 
(Run Procedure) outputs $t$ if and only if there exists a switching flow to
$t$. 
\end{theorem}

The switching flow corresponding to the actual train run can be
characterized as follows.

\begin{theorem}[The run profile is the minimal switching
  flow~\cite{Dohrau2017}]\label{thm:run_profile}
  Let $A=(V, \, o, \, d, \, \dbar, \, s_{even}, \, s_{odd})$ be a
  terminating \arr{} instance with edges $E$. Let $\hat{x}$ be the
  \emph{run profile} of $A$, meaning that $\hat{x}_e$ counts the number of
  times edge $e$ is traversed during
  Algorithm~\ref{algo:run_procedure} (Run Procedure).
  Then $\hat{x}\leq x$ for all switching flows $x$. In particular,
  $\hat{x}$ is the unique minimizer of the total flow
  $\sum_{e\in E} x_e$ over all switching flows.
\end{theorem}

We note that this provides the missing direction of
Theorem~\ref{theorem:dohrau}. Indeed, $\hat{x}$ is a switching flow
and hence either has $\hat{x}^-(d)=1$ or $\hat{x}^-(\dbar)=1$. By
$\hat{x}\leq x$, every switching flow $x$ is to the same
destination. In general, there can be switching flows $x\neq \hat{x}$~\cite{Dohrau2017}.

We will derive Theorem~\ref{thm:run_profile} as a special
case of Theorem~\ref{thm:unique_x} in the next section.

\section{A GENERAL FRAMEWORK}\label{sec:framework}
In order to solve the \arr\ problem, we can simulate
Algorithm~\ref{algo:run_procedure} (Run Procedure) which takes
exponential time in the worst case~\cite{Dohrau2017}; alternatively,
we can try to get hold of a switching flow; via
Theorem~\ref{theorem:dohrau}, this also allows us to decide \arr.

According to Definition~\ref{def:switching_flow}, a switching flow can
be obtained by finding a feasible solution to an integer linear
program (ILP); this is a hard task in general, and it is unknown
whether switching flow ILPs can be solved more efficiently than
general ILPs.

In this section, we develop a framework that allows us to
reduce the problem to that of solving a number of more constrained ILPs. 
At the same time, we provide direct methods for solving them that do 
not rely on using general purpose ILP solvers.

\subsection{The idea}
Given a terminating \arr\ instance, we
consider the switching flow conditions in
Definition~\ref{def:switching_flow}. Given an arbitrary fixed subset $S=\{v_1,\ldots,v_k\}\subseteq V$ of $k$ vertices, we drop the flow
conservation constraints at the vertices in $S$, but at the same time
prescribe outflow values $x^+(v_1),\ldots, x^+(v_k)$ that we can think
of as guesses for their values in a switching flow. 

If we minimize the total flow subject to these guesses, we 
obtain a unique solution (Theorem~\ref{thm:unique_x}~(i) below) and hence
unique inflow values $x^-(v_1),\ldots, x^-(v_k)$ for the vertices in
$S$. If we happen to stumble upon a fixed point of the mapping
$x^+(v_1),\ldots, x^+(v_k)\rightarrow x^-(v_1),\ldots, x^-(v_k)$, we
recover flow conservation also at $S$, which means that our guesses were
correct and we have obtained a switching flow.

The crucial property is that the previously described mapping is
\emph{monotone} (Theorem~\ref{thm:unique_x}~(ii) below), meaning that the theory of
\emph{Tarski fixed points} applies that guarantees the existence of a
fixed point as well as efficient algorithms for finding it (Lemma~\ref{lem:tarski} below).

Hence, we reduce the computation of a switching flow to a benign search
problem (for a Tarski fixed point), where every search step requires
us to solve a ``guessing'' ILP. We next present a ``rail'' way of solving the
guessing ILP that turns out to be more efficient in the worst case (and also simpler)
than general purpose ILP solvers. For suitable switch graphs and
appropriate choices of the set $S$, it will be fast enough to yield
the desired runtime results.  

\subsection{The Multi-Run Procedure}\label{sec:multirun}
Given $S=\{v_1,\ldots,v_k\}\subseteq V$ and $w\in\N_0^k$ (guesses for
the outflows from the vertices in $S$), we start one train from $Y$
and $w_i$ trains from $v_i$ until they arrive back in $S$, or at a
destination. In this way, we produce inflow values for the vertices in
$S$.

By starting, we mean that we move each of the trains by one step: the
one on $Y$ moves to $o$, while $\ceil{w_i/2}$ of the ones at $v_i$
move to the even successor of $v_i$, and $\floor{w_i/2}$ to the odd
successor. Trains that are now on vertices in $V\setminus S$ are called
\emph{waiting} (to move on).

For all $v\in V\setminus S$, we initialize current and next
successors as before in Algorithm~\ref{algo:run_procedure} (Run Procedure). Then we
(nondeterministically) repeat the following until there are no more trains waiting.

We pick a vertex $v\in V\setminus S$ where some trains are waiting
and call the number of waiting trains $t(v)$. We choose a number
$\tau\in\{1,\ldots,t(v)\}$ of trains to move on; we move
$\ceil{\tau/2}$ of them to the current successor and $\floor{\tau/2}$
to the next successor. If $\tau$ is odd, we afterwards swap the
current and the next successor at $v$.

Algorithm~\ref{algo:multi_run_procedure} (Multi-Run Procedure)
provides the details. For $S=\emptyset$, the procedure becomes
deterministic and is equivalent to Algorithm~\ref{algo:run_procedure}
(Run Procedure).

\begin{algorithm}[htb]
\DontPrintSemicolon
\SetKwRepeat{Do}{do}{while}
\KwIn{Terminating \arr{} instance  $A=(V, \, o, \,d, \, \dbar, \,
  s_{even}, \, s_{odd})$ with edges $E$; \\
$S =\{v_1, \, v_2, \, \dots, \, v_k \}\subseteq V$, $w = (w_1,
\, w_2, \, \dots, \, w_k) \in \N_0^k$ (one train starts from $Y$, and $w_i$ trains start from $v_i$).}
\KwOut{number of trains arriving at $d, \dbar$, and in $S$, respectively}

Let $t$ be a zero-initialized array indexed by the vertices of $V\cup\dest$ \;
$t[o] \gets 1$ \tcc*[r]{traversal of $(Y,o)$}
\For{$i=1,2,\ldots,k$}{
  $t[s_{even}(v_i)] \gets t[s_{even}(v_i)] + \ceil{w_i/2}$
   \tcc*[r]{$\ceil{w_i/2}$ traversals of $(v_i, s_{even}(v_i))$} 
   $t[s_{odd}(v_i)] \gets t[s_{odd}(v_i)] + \floor{w_i/2}$
   \tcc*[r]{$\floor{w_i/2}$ traversals of $(v_i, s_{odd}(v_i))$} 
}
Let $s_{curr}$ and $s_{next}$ be arrays indexed by the vertices of
$V\setminus S$ \;
\For{$v \in V\setminus S$}{
	$s_{curr}[v] \gets s_{even}(v)$ \;
	$s_{next}[v] \gets s_{odd}(v)$ \;
}

\While{$\exists v\in V\setminus S: t[v]>0$}{
  pick $v\in V\setminus S$ such that $t[v]>0$ and choose $\tau\in\{1,\ldots,t[v]\}$\;
  $t[v] \gets t[v] - \tau$ \;
  $t[s_{curr}(v)] \gets t[s_{curr}(v)] + \ceil{\tau/2}$
      \tcc*[r]{$\ceil{\tau/2}$ traversals of $(v, s_{curr}(v))$} 
   $t[s_{next}(v)] \gets t[s_{next}(v)] + \floor{\tau/2}$
       \tcc*[r]{$\floor{\tau/2}$ traversals of $(v, s_{next}(v))$}
   \If{$\tau$ is odd}{
      swap($s_{curr}[v], \, s_{next}[v])$\;
    }
  }
  \Return{$(t[d],t[\dbar],t[v_1], t[v_2],\ldots,t[v_k])$}
\caption{Multi-Run Procedure}
\label{algo:multi_run_procedure}
\end{algorithm}

\begin{lemma}\label{m_termination}
  Algorithm~\ref{algo:multi_run_procedure} (Multi-Run Procedure)
  terminates.
\end{lemma}

\begin{proof}
  This is a qualitative version of the argument in
  Lemma~\ref{lem:visits}.  Let $x:E\rightarrow\N_0$ record how many
  times each edge $e\in E$ has been traversed in total, at any given
  time of Algorithm~\ref{algo:multi_run_procedure} (Multi-Run
  Procedure). For $v\in V\setminus S$, we always have
  $x^+(v)=x^-(v)-t(v)$, where $t(v)$ is the number of trains currently
  waiting at $v$. Suppose for a contradiction that the Multi-Run
  procedure cycles. Then $x^-(v)$ is unbounded for at least one
  $v\in V\setminus S$, which means that $x^+(v)$ is also unbounded,
  since $t(v)$ is bounded. This in turn means that $x^-(s_{even}(v))$
  and $x^-(s_{odd}(v))$ are unbounded as well, since we distribute
  $x^+(v)$ evenly between the two successors. Repeating this argument,
  we see that $x^-(w)$ is unbounded for all vertices $w$ reachable from
  $v$. But as $x^-(d)$ and $x^-(\dbar)$ are bounded (by the number of
  trains that we started), neither $d$ nor $\dbar$ are reachable from
  $v$. This is a contradiction to $A$ being terminating.
\end{proof}

\subsection{Candidate switching flows}
After Algorithm~\ref{algo:multi_run_procedure} (Multi-Run Procedure)
has terminated, let $\hat{x}_e$ be the number of times the edge $e$ was
traversed.  We then have flow conservation at $v\in V\setminus S$,
switching behavior at $v\in V$ and outflow $w_i$ from $v_i$. Indeed,
every train that enters $v\in V\setminus S$ eventually also leaves it;
moreover, the procedure is designed such that it simulates moving
trains out of $v\in V$ individually, strictly alternating between
successors. Finally, as we start $w_i$ trains from $v_i \in S$ and stop all
trains once they arrive in $S$, we also have outflow $w_i$ from $v_i$.

We remark that we do not have any control over how many trains end up
at $d$ or $\dbar$. Also, $\hat{x}$ could in principle depend on the order in
which we pick vertices, and on the chosen $\tau$'s. We will show in
Theorem~\ref{thm:unique_x} below that it does not. So far, we have only
argued that $\hat{x}$ is a \emph{candidate switching flow} according to the
following definition.

\begin{definition}[Candidate Switching Flow]\label{def:cand_switching_flow}
  Let $A=(V, \, o, \,d, \, \dbar, \, s_{even}, \, s_{odd})$ be a
  terminating \arr{} instance with edges $E$, 
  $S =\{v_1, \, v_2, \, \dots, \, v_k \}\subseteq V$, 
  $w = (w_1, \, w_2, \, \dots, \, w_k) \in \N_0^k$.

   A function $x: E \rightarrow \N_0$ is a \emph{candidate switching
     flow} for $A$ (w.r.t.\ $S$ and $w$) if
\begin{equation}\label{eq:cand_flow}
\begin{array}{rccll}
x^+(Y) &=& 1, \\
  x^+(v) - x^-(v) &=& 0, & \quad v\in V\setminus S  & \text{(flow
                                                      conservation at
                                                      $V\setminus S$)}\\
  x^+(v_i) &=& w_i, & \quad i=1,2,\ldots,k, & \text{(outflow $w$ at $S$)}\\
  x_{(v, s_{even}(v))} -x_{(v, s_{odd}(v))}&\in&\{0,1\}, & \quad v\in
                                                           V&
                                                              \text{(switching
                                                              behavior)}.
\end{array}
\end{equation}
\end{definition}

\begin{theorem}[Each Multi-Run profile is the minimal candidate switching
  flow]\label{thm:unique_x}
  Let $A,E,S,w$ be as in Definition~\ref{def:cand_switching_flow} and let $\hat{x}$ be a
  \emph{Multi-Run profile} of $A$, meaning that $\hat{x}_e$ is the number of
  times edge $e\in E$ was traversed during some run of 
  Algorithm~\ref{algo:multi_run_procedure} (Multi-Run Procedure).
  Then the following statements hold.
  \begin{itemize}
\item[(i)]  $\hat{x}\leq x$ for all candidate switching flows $x$ (w.r.t.\
  $S$ and $w$). In particular,
  $\hat{x}$ is the unique minimizer of the total flow
  $\sum_{e\in E} x_e$ over all candidate switching flows.
\item[(ii)] For fixed $A, E, S$, define
      $F(w) = (\hat{x}^-(v_1),\ldots, \hat{x}^-(v_k))\in \N_0^k$. Then
      the function $F: \N_0^k\rightarrow \N_0^k$ is \emph{monotone},
      meaning that $w\leq w'$ implies that $F(w)\leq F(w')$.
\end{itemize}
\end{theorem}

\begin{proof}
  We prove part (i) by the \emph{pebble argument}~\cite{Dohrau2017}:
  Let $x$ be any candidate switching flow w.r.t.\ $S$ and $w$.  For
  every edge $e$, we initially put $x_e$ pebbles on $e$, and whenever
  a train traverses $e$ in Algorithm~\ref{algo:multi_run_procedure}
  (Multi-Run Procedure), we let it collect
  a pebble. If we can show that we never run out of pebbles,
  $\hat{x}\leq x$ follows. By ``running out of pebbles'', we concretely
  mean that we are for the first time trying to collect a pebble from
  an edge with no pebbles left.

  Since $x$ is a candidate switching flow, we cannot run out of
  pebbles while starting the trains. In fact, we exactly collect all
  the pebbles on the outgoing edges of $\{Y\}\cup S$. It remains to
  show that we cannot run out of pebbles while processing a picked vertex
  $v\in V\setminus S$. For this, we prove that we maintain the
  following additional invariants (which hold immediately after starting the
  trains). Let $p:E\rightarrow \N_0$ record for each edge $e$ the
  remaining number of pebbles on $e$. Then for all $v\in V\setminus S$,
  \begin{itemize}
  \item[(a)] $p^+(v)= p^-(v) + t(v)$, where $t(v)$ is the number of trains
    waiting at $v$; 
  \item[(b)] $p ((v, s_{curr}(v))) - p ((v, s_{next}(v)))\in\{0,1\}$.
  \end{itemize}
  Suppose that these invariants hold when picking a vertex $v\in V\setminus
  S$. As we have not run out of pebbles before, $p^-(v)\geq 0$ and (a)
  guarantees that we have $q\geq t(v)$ pebbles on the outgoing edges;
  by (b), $\ceil{q/2}$ of them are on $(v, s_{curr}(v))$ and
  $\floor{q/2}$ on $(v, s_{next}(v))$. From the former, we collect
  $\ceil{\tau/2}$, and from the latter $\floor{\tau/2}$ where
  $\tau\leq t(v)\leq q$, so we do not run out of pebbles. We maintain (a) at
  $v$ where both $p^+$ and $t$ are reduced by $\tau$.  We
  also maintain (a) at the successors; there, the
  gain in $t$ exactly compensates the loss in $p^-$. Finally, we
  maintain (b) at $v$: If $\tau$ is even, both $p ((v, s_{curr}(v)))$
  and $p ((v, s_{next}(v)))$ shrink by $\tau/2$. If $\tau$ is
  odd, we have $p ((v, s_{curr}(v))) - p ((v, s_{next}(v)))\in\{-1,0\}$
  after collecting one more pebble from $(v, s_{curr}(v))$
  than from $(v, s_{next}(v))$, but then we reverse the sign
  by swapping $s_{curr}$ and $s_{next}$.

  For $S=\emptyset$, this proves Theorem~\ref{thm:run_profile}, and
  for general $S$, we have now proved (i). In particular, the order in
  which we move trains in Algorithm~\ref{algo:multi_run_procedure} (Multi-Run Procedure) does not matter.

  The proof of (ii) is now an easy consequence; recall that the inflow
  $F(w)_i$ is the number of trains that arrive at $v_i$. If
  $w\leq w'$, we run Algorithm~\ref{algo:multi_run_procedure} (Multi
  Run Procedure) with input $w'$ such that it first simulates a run
  with input $w$; for this, we keep the extra trains corresponding to
  $w'-w$ waiting where they are after the start, until all other
  trains have terminated. At this point, we have inflow $f\geq F(w)$
  at $S$, where $f-F(w)$ corresponds to the extra trains that have
  already reached $S$ right after the start. We finally run the extra trains
  that are still waiting, and as this can only further increase the
  inflows at $S$, we get $F(w')\geq f\geq F(w)$.
\end{proof}

\subsection{Runtime}

As we have proved in Theorem~\ref{thm:unique_x}~(i), the Multi-Run
procedure always generates the unique flow-minimal candidate switching
flow. But the number of steps depends on the order in which vertices
$v\in V\setminus S$ are picked, and on the chosen $\tau$'s. We start
with an upper bound on the number of edge traversals that generalizes
Lemma~\ref{lem:termination}. 

\begin{lemma}\label{lem:m_termination}
  Let $A=(V, \, o, \,d, \, \dbar, \, s_{even}, \, s_{odd})$ be a
  terminating \arr{} instance, $|V|=n$,
  $S =\{v_1, \, v_2, \, \dots, \, v_k \}\subseteq V$,
  $w = (w_1, \, w_2, \, \dots, \, w_k) \in \N_0^k$.
  Let $\ell$ be the maximum length of the shortest path from a vertex
  in $V\setminus S$ to a vertex in $\dest \cup S$.
  Further suppose that at the beginning of some iteration in
  Algorithm~\ref{algo:multi_run_procedure} (Multi-Run Procedure), $R$
  trains are still waiting. Then all subsequent iterations traverse at
  most $R((n - \ell + 2) 2^\ell-2)$ edges in total.
\end{lemma}

\begin{proof}
  We continue to run each of the $R$ waiting trains individually and
  proceed with the next one only when the previous one has
  terminated. In Algorithm~\ref{algo:multi_run_procedure} (Multi-Run
  Procedure), this corresponds to always using $\tau=1$ and the next
  vertex $v$ as the head of the previously traversed edge, for each of
  the $R$ trains. So we effectively perform Algorithm~\ref{algo:run_procedure} (Run
  Procedure) for $R$ trains.

  As each train terminates once it reaches a vertex in $S\cup\dest$,
  Lemmata~\ref{lem:visits} and~\ref{lem:termination} are easily seen
  to hold also here, after redefining ``destination'' as any vertex in
  $S\cup\dest$. As a consequence, each train traverses at most
  $(n - \ell + 2) 2^\ell-2$ edges until it reaches a vertex in
  $\dest \cup S$.  This leads to at most $R((n - \ell + 2) 2^\ell-2)$
  edge traversals overall. By Theorem~\ref{thm:unique_x}~(i),
  this upper bound holds for all ways of continuing
  Algorithm~\ref{algo:multi_run_procedure} (Multi-Run Procedure).
\end{proof}

With $R=W:=1+\sum_{i=1}^k w_i$, we obtain an upper bound
for the total number of loop iterations since each iteration traverses at
least one edge. But it turns out that we can be significantly faster (and polynomial in
the encoding size of $W$) when we proceed in a greedy fashion, i.e.\
we always pick the next vertex as the one with the largest number of
waiting trains, and move all these trains at once.

\begin{lemma}\label{lem:fast_m_termination}
  Let $A, n, S, w, \ell$ as in Lemma~\ref{lem:m_termination}, and suppose
  that in each iteration of Algorithm~\ref{algo:multi_run_procedure}
  (Multi-Run Procedure), we pick $v\in V\setminus S$ maximizing $t[v]$
  and further choose $\tau=t[v]$.  Then the number of iterations is at
  most $(\ln W+n)(n-k)((n - \ell + 2) 2^\ell-2)$, where $W = 1+ \sum_{i=1}^k w_i$.
\end{lemma}

\begin{proof}
  As in the proof of Theorem~\ref{thm:unique_x}, we let each train
  collect a pebble as it traverses an edge, where we initially put
  $\hat{x}_e$ pebbles on edge $e$, with $\hat{x}$ being the unique Multi-Run
  profile. This means that we eventually collect all pebbles. Now
  consider an iteration and suppose that $R\leq W$ trains are still
  waiting. In the greedy algorithm, we move at least $R/(n-k)$ of them in this
  iteration and collect at least that many pebbles. On the other
  hand, with $R$ trains still waiting, and with $T=(n - \ell + 2) 2^\ell-2$, there can be no more than
  $RT$ pebbles left, as all of them will be collected in the remaining
  at most that many edge traversals, due to
  Lemma~\ref{lem:m_termination}.

  In summary, the number of pebbles is guaranteed to be reduced by a factor of
  \[
    \left (1-\frac{1}{(n-k)T}\right)
   \]
   in each iteration, starting from at most $WT$ pebbles before the
   first iteration. After $s=(\ln W+n) (n-k)T$ iterations, we
   therefore have at most 
   \[
     \left (1- \frac{1}{(n-k)T}\right)^s WT \leq e^{-\ln W-n} WT<1
   \]
   pebbles left (using
   $T<e^n$). Hence, after at most $s$ iterations, the greedy version of
   Algorithm~\ref{algo:multi_run_procedure} (Multi-Run Procedure) has
   indeed terminated. 
 \end{proof}

 We remark that essentially the same runtime can be achieved by a
 round robin version that repeatedly cycles through $V\setminus S$
 in some fixed order.

\subsection{Tarski fixed points}\label{sec:tarski}
Tarski fixed points arise in the study of order-preserving functions
on complete lattices~\cite{Tarski1955}. For our application, it
suffices to consider finite sets of the form
$L = \{0, \, 1, \, \dots, \, N\}^k$ for some $N, k \in \N^+$. For such
a set, Tarski's fixed point theorem~\cite{Tarski1955} states that any
monotone function $D : L \rightarrow L$ has a fixed point, some
$\hat{w}\in L$ such that $D(\hat{w})=\hat{w}$. Moreover, the problem of finding such a
fixed point has been studied: Dang, Qi and Ye~\cite{Dang2020} have
shown that a fixed point can be found using $\bigO(\log^k N)$
evaluations of $D$. Recently, Fearnley, P\'alv{\"o}lgyi and Savani~\cite{Fearnley2020}
improved this to $\bigO(\log^{2\ceil{k/3}} N)$.

Via Theorem~\ref{thm:unique_x}, we have reduced the problem of
deciding a terminating \arr\ instance to the problem of finding a
fixed point of a monotone function $F:\N_0^k\rightarrow\N_0^k$,
assuming that we can efficiently evaluate $F$. Indeed, if we have such
a fixed point, the corresponding (flow-minimal) candidate switching flow
is an \emph{actual} switching flow and hence decides the problem via
Theorem~\ref{theorem:dohrau}.

The function $F$ depends on a set $S\subseteq V$ of size $k$ that
we can choose freely (we will do so in the subsequent sections).

Here, we still need to argue that we can restrict $F$ to a finite set
$L=\{0, \, 1, \, \dots, \, N\}^k$ so that the Tarski fixed point
theorem applies. We already know that outflow (and hence inflow)
values never exceed $N=2^n$ in \emph{some} switching flow, namely the run profile
(Lemma~\ref{lem:visits}), so we simply restrict $F$ to this range
and at the same time cap the function values accordingly.

\begin{lemma}\label{lem:tarski}
Let $A=(V, \, o, \, d, \, \dbar, \, s_{even}, \, s_{odd})$ be a
  terminating \arr{} instance, $S=\{v_1,\ldots,v_k\}\subseteq
  V$, $|V|=n$. Let $F$ be the function defined in
  Theorem~\ref{thm:unique_x}~(ii), let $N=2^n$ and consider the function
  $D: \{0, \, 1, \, \dots, \, N\}^k\rightarrow \{0, \, 1, \, \dots, \,
  N\}^k$ defined by
  \[
    D(w) = \left(\begin{array}{l}
                   \min (N, F(w)_1)\\
                   \min (N, F(w)_2)\\
                   \vdots \\
                    \min (N, F(w)_k)
                 \end{array}
               \right), \quad w\in \{0, \, 1, \, \dots, \, N\}^k.\]
Then $D$ is monotone and has a fixed point $\hat{w}$ that can be found with
  $\bigO(\log^{2\ceil{k/3}} N)$ evaluations of $D$. Moreover, $\hat{w}$ is also a fixed
  point of $F$, and when we apply Theorem~\ref{thm:unique_x}~(i) with $w=\hat{w}$,
  the flow-minimal candidate switching flow resulting from the
  multi-run procedure is a switching flow for $A$.
\end{lemma}

We remark that the switching flow obtained in this way is not
necessarily flow-minimal, so we cannot argue that we obtain the run
profile of $A$ as defined in Theorem~\ref{thm:run_profile}. The function $D$
may have several fixed points, each of them leading to a different
switching flow; to obtain the run profile, we would have to find a
particular fixed point, the one that leads to the unique switching flow of
smallest total flow. The known Tarski fixed point algorithms
cannot do this, and we do not know of any efficient method for computing
the run profile from a given switching flow.

\begin{proof}
  Monotonicity is clear: if $w\leq w'$, then $F(w)\leq F(w')$ by monotonicity of $F$; see Theorem~\ref{thm:unique_x}~(ii). But then also 
  $D(w)\leq D(w')$ for the capped values. Hence, the Tarski fixed
  point theorem~\cite{Tarski1955} yields a fixed point $\hat{w}$ of
  $D$, and the algorithm of Fearnley, P\'alv{\"o}lgyi and Savani~\cite{Fearnley2020} finds it
  using $\bigO(\log^{2\ceil{k/3}} N)=\bigO(n^{2\ceil{k/3}})$ evaluations.
  

  It remains to prove that $\hat{w}$ is a fixed point of $F$. Suppose
  for a contradiction that it is not a fixed point. Then
  $F(\hat{w})\neq D(\hat{w})$, i.e.\ some values were actually capped, and so
  $\hat{w}_j=D(\hat{w})_j=N<F(\hat{w})_j$ for at least one $j$. As we
  also have $\hat{w}=D(\hat{w})\leq F(\hat{w})$, we get
  \begin{equation}\label{eq:Dproof1}
    \sum_{i=1}^k \hat{w}_i < \sum_{i=1}^k F(\hat{w})_i. 
  \end{equation}
  On the other hand, consider the candidate switching flow 
  (\ref{eq:cand_flow}) with $w=\hat{w}$. At most the total flow emitted (at $Y$
  and the $v_i$'s) is absorbed at $S$, so we have
 \begin{equation}\label{eq:Dproof2}
   \sum_{i=1}^k F(\hat{w})_i  \leq 1+\sum_{i=1}^k \hat{w}_i.
 \end{equation}
 Putting this together with (\ref{eq:Dproof1}), we get an equality in
 (\ref{eq:Dproof2}). In particular, $v_j$ is the only vertex whose
 inflow value was capped (by one), all emitted flow is absorbed at
 $S$, and no flow arrives at $d$ or $\dbar$.

  But this is a contradiction to $\hat{w}_j= N = 2^n$: By the same
  arguments as in the proof of Lemma~\ref{lem:visits}, based on flow
  conservation (at all $v\neq v_j$) and switching behavior, one of
  these $2^n$ outflow units is guaranteed to arrive at $\dest$.
\end{proof}

\section{SUBEXPONENTIAL ALGORITHM FOR \arr}\label{sec:subex}
In this section, we present our main application of the general
framework developed in the previous section.

Given a terminating \arr\ instance $A$ with $|V|=n$, the plan is to
construct a set $S\subseteq V$ of size $\bigO(\sqrt{n})$ such that
from any vertex, the length of the shortest path in $G(A)$ to a vertex
in $S\cup\dest$ is also bounded by roughly $\bigO(\sqrt{n})$. Since $S$ is
that small, we can find a Tarski fixed point with a subexponential number
of $F$-evaluations; and since shortest paths are that short, each
$F$-evaluation can also be done in subexponential time using the
Multi-Run procedure. An overall subexponential algorithm ensues.

\begin{lemma}
\label{lem:phi-set}
  Let $A=(V, \, o, \, d, \, \dbar, \, s_{even}, \, s_{odd})$ be a
  terminating \arr{} instance with $|V| = n$. Let $\phi\in(0,1)$ be a
  real number. In $\bigO(n)$ time, we can construct a \emph{$\phi$-set}
  $S$, meaning a set $S\subseteq V$ such that
  \begin{itemize}
  \item[(i)] $|S| \leq \phi\cdot(n+2)$;
  \item[(ii)] for all $v \in V$, the shortest path from
    $v$ to $S\cup\dest$ in $G(A)$ has length at most $\log_2(n+2) / \phi$.
  \end{itemize}
\end{lemma} 
\begin{proof}
We adapt the ball-growing technique of Leighton and Rao~\cite{Leighton1999}, as explained by Trevisan~\cite{Trevisan2005}. 

We first decompose the switch graph $G(A)$ into layers based on the
distance of the vertices to a destination~\cite{Hung}. 
More formally, for $v \in V \cup \dest$, we denote by $\dist(v)$ the length of the shortest path from $v$ to $\dest$ in $G(A)$.
Then the layers are defined as $L_i := \{v \in V \cup \dest: \dist(v) = i\}$ for $i \geq 0$.
Define $\ell := \max\{\dist(v) : v \in V \}$.
We can compute the layer decomposition $(L_0, \dots, L_{\ell})$ using
breadth-first search in $\bigO(n)$ time.

Consider the following procedure that computes a $\phi$-set as a union of layers:

\begin{algorithm}[htb]
\DontPrintSemicolon
\SetKwRepeat{Do}{do}{while}
\SetKwComment{cmt}{(*}{*)}
\SetNoFillComment 
\KwIn{\arr{} instance with layer decomposition $(L_0, \dots, L_{\ell})$, $\phi \in (0,1)$}
\KwOut{a $\phi$-set $S$}
$S\gets \emptyset$ \; 
$U\gets L_0$ \;
\For{$i = 1, \dots, \ell$}{
  \If{$|L_i| < \phi |U|$}{
        $S\gets S \cup L_i$\;
	$U\gets  \emptyset$ \;

    } 
    $U\gets  U \cup L_i$\;
}
\Return{$S$}
\caption{Procedure to compute a $\phi$-set}
\end{algorithm}

It is clear that the procedure is done in $\bigO(n)$ time. To prove (i),
we observe that whenever we add a layer $L_i$ to $S$, we have
$|L_i|<\phi |U|$; moreover, the $U$'s considered in these inequalities
are mutually disjoint subsets of $V\cup\dest$. Hence, $|S|<\phi\cdot(n+2)$.

For (ii), let $v\in V$. Then  $v\in L_b$ for some $b\geq 1$. Let
$0\leq a\leq b$ be the largest index such that $L_a\subseteq
S\cup\dest$. Then the shortest path from $v$ to a vertex in
$S\cup\dest$ has length at most $b-a$. It remains to bound $j:=b-a$. The
interesting case is $j>0$.

Consider the above algorithm. After the $a$-th iteration, we have
$|U|=|L_a|\geq 1$. Moreover, $|L_i|\geq \phi|U|$ for $i=a+1,\ldots,b$, meaning that for each
iteration $i$ in this range, the size of $U$ has grown by a factor of at
least $1+\phi$. Hence, after the $b$-th
iteration, $(1+\phi)^{j}\leq |U|\leq n+2$.  This implies $j \leq \log_2(n+2) /
\log_2(1+\phi) < \log_2(n+2) / \phi$, where we use the inequality
$\log_2(1 + \phi) > \phi$ for $\phi \in (0,1)$.
\end{proof}

\begin{theorem}
  Let $A=(V, \, o, \, d, \, \dbar, \, s_{even}, \, s_{odd})$ be a terminating \arr{} instance with $|V| = n$.
  $A$ can be decided in time $\bigO(p(n) n^{1.633\sqrt{n}})$, for some polynomial $p$.
\end{theorem}

\begin{proof}
  By Lemma~\ref{lem:phi-set}, we can find a $\phi$-set $S$ in $\bigO(n)$ time, for any $\phi \in (0,1)$.
  As $|S| \leq \phi \cdot (n+2)$, by Lemma~\ref{lem:tarski}, we can then decide $A$ with $\bigO(n^{2\lceil\phi \cdot (n+2)/3\rceil})$ evaluations of the function $D$. 
  Each evaluation in turn requires us to evaluate the function $F$ in Theorem~\ref{thm:unique_x}~(ii) for a given $w\in\{0,1,\ldots,2^n\}^{|S|}$.
  We can do this by applying Algorithm~\ref{algo:multi_run_procedure} (Multi-Run Procedure).
  By Lemma~\ref{lem:fast_m_termination} and the definition of a $\phi$-set in Lemma~\ref{lem:phi-set}, running this algorithm in a greedy fashion requires at most $(\ln W+n)(n-|S|)((n-\ell +2)2^\ell - 2)$ iterations, where $W = 1+ \sum_{i=1}^{|S|} w_i$ and $\ell = \log_2(n+2) / \phi$.
  Further, from the choice of $w$, we have $W \leq 2^n \phi (n+2) + 1$.
  Therefore, the number of iterations is $\bigO(q(n) n^{1/\phi})$ for some polynomial $q$.
  At each iteration, we need to find the vertex with the highest number of waiting trains, as stated in Lemma~\ref{lem:fast_m_termination}, and move the trains from the chosen vertex.
  All these operations take polynomial time.
  
  In total, the runtime of the whole process is $\bigO(n^{2\lceil\phi \cdot (n+2)/3\rceil}\cdot p(n) n^{1 / \phi})$ for some polynomial $p$.
  Choosing $\phi = \sqrt{3}/\sqrt{2n}$, the runtime becomes $\bigO(p(n) n^{1.633\sqrt{n}})$.
\end{proof}

\section{FEEDBACK VERTEX SETS}\label{sec:feedback}
In the previous section, we used our framework to obtain an improved algorithm for \arr\ in general. In this section, we will instantiate the framework differently to obtain a polynomial-time algorithm for a certain subclass of \arr.

A subset $S \subseteq V$ of vertices in a directed graph
$G = (V, \, E)$ is called a \emph{feedback vertex set} if and only if
the subgraph induced by $V \setminus S$ is acyclic (i.e.\ it contains
no directed cycle). Karp~\cite{Karp1972} showed that the problem of
finding a smallest feedback vertex set is $\NP$-hard. However, there
exists a parameterized algorithm by Chen et al.\@~\cite{Chen2008} which
can find a feedback vertex set of size $k$ in time
$\bigO(n^44^kk^3k!)$ in a directed graph on $n$ vertices, or report
that no such set exists.

If we apply Theorem~\ref{thm:unique_x} with a feedback vertex set
$S$, it turns out that we can compute the Multi-Run profile in
polynomial time, meaning that we get a polynomial-time algorithm for
\arr\ if there is a feedback vertex set of constant size $k$.

\begin{theorem}
  Let $A=(V, \, o, \, d, \, \dbar, \, s_{even}, \, s_{odd})$ be a
  terminating \arr{} instance with graph $G(A)$. If $G(A)$ has a
  feedback vertex set $S\subseteq V$ of size $k$ (assumed to be
  fixed as $n=|V|\rightarrow\infty$), then $A$ can be decided in
  time $\bigO(n^4 + n^{2(\lceil k/3\rceil+1)})$.
\end{theorem}

\begin{proof}
  Using the algorithm by Chen et al.\@~\cite{Chen2008}, we can find a
  feedback vertex set $S$ in $\bigO(n^4)$ time if it exists.  According to
  Lemma~\ref{lem:tarski}, we can then decide $A$ with
  $\bigO(n^{2\lceil k/3\rceil})$
  evaluations of the function $D$. Each evaluation in turn requires us
  to evaluate the function $F$ in Theorem~\ref{thm:unique_x}~(ii) for
  a given $w\in\{0,1,\ldots,2^n\}^k$. To do this, we apply
  Algorithm~\ref{algo:multi_run_procedure} (Multi-Run
  Procedure) where we pick vertices $v\in V\setminus S$ in topological
  order and choose $\tau=t[v]$ always. As we never send any
  trains back to vertices that have previously been picked, we 
  terminate within $n-k$ iterations, each of which can
  be performed in time $\bigO(n)$ as it involves $\bigO(n)$-bit
  numbers. Hence, $F(w)$ can be computed in $\bigO(n^2)$ time. 
  The claimed runtime follows.
\end{proof}

We remark that even if $k$ is not constant, we can still beat the
subexponential algorithm in Section~\ref{sec:subex}, as long as
$k=\bigO(n^{\alpha})$ for some $\alpha<1/2$.

\subsection*{Acknowledgment} We thank G\"unter Rote for pointing out
an error in an earlier version of the manuscript.

\bibliography{references}
\bibliographystyle{acm}

\end{document}